\documentclass[envcountsame]{llncs}
\usepackage[ruled, vlined]{algorithm2e}
\usepackage{amsfonts,amssymb,amstext,amsmath}
\usepackage{tikz}
\usetikzlibrary{arrows, petri, topaths, automata}
\usepackage{tkz-berge}
\usepackage[textsize=tiny]{todonotes}
\usepackage{xspace}
\usepackage{thmtools,thm-restate}
\usepackage{graphicx}
\usepackage{enumerate}
\usepackage{floatrow}
\usepackage{fancyvrb}
\usepackage{a4wide}
\usepackage{booktabs}
\usepackage{xspace}
\usepackage{mathtools}
\usepackage{amsmath}
\usepackage{amsfonts}
\usepackage[ruled]{algorithm2e}

\newcommand{\NP}{NP}

\newcommand{\ea}{$(1+1)$~\mbox{EA}\xspace}
\newcommand{\gsemo}{\mbox{\textsc{Gsemo}}\xspace}
\newcommand{\bigo}[1]{\mathcal{O}\left(#1 \right)}
\newcommand{\ds}{\textsf{MDS}\xspace}
\newcommand{\vc}{\textsf{MVC}\xspace}
\newcommand{\is}{\textsf{MIS}\xspace}
\newcommand{\cds}{\textsf{CDS}\xspace}

\newcommand{\degr}{\delta}
\newcommand{\opt}{\textsc{opt}\xspace}
\newcommand{\app}{$\epsilon$-approximation\xspace}
\newcommand{\absl}[1]{\left \lvert #1 \right \rvert}

\newcommand{\optd}{D}
\newcommand{\optc}{C}

\newcommand{\optm}{M_*}

\newcommand{\ignore}[1]{}

\newcommand{\abeq}{\begin{equation} a  := \frac{\beta - 1}{\beta - 2} \left [1 - \left (\frac{t + 2}{t + 1}\right )^{1 - \beta} \right ]^{-1} \quad \mbox{and}\quad b := \left [ 4 c_1 \frac{(t + 1)^{\beta - 1}}{\beta - 1} \right ]^{\frac{1}{\beta - 2}}. \nonumber \end{equation}}

\DeclareMathOperator*{\argmax}{arg\,max}
\DeclareRobustCommand{\PLBU}{PLB\xspace}

\renewenvironment{proof}{\paragraph{Proof}}{\hfill$\square$\vspace{10pt}\\}

\begin{document}

\title{Approximating Optimization Problems using EAs on Scale-Free Networks}

\titlerunning{Approximating Optimization Problems using EAs on Scale-Free Networks} 

\author{Ankit Chauhan \and Tobias Friedrich \and Francesco Quinzan}

\authorrunning{Chauhan et al.}
\tocauthor{Ivar Ekeland, Roger Temam, Jeffrey Dean, David Grove,
Craig Chambers, Kim B. Bruce, and Elisa Bertino}
\institute{Hasso Plattner Institute, Potsdam, Germany}

\maketitle         

\begin{abstract}
It has been observed that many complex real-world networks have certain properties, such as a high clustering coefficient, a low diameter, and a power-law degree distribution. A network with a power-law degree distribution is known as \emph{scale-free network}. In order to study these networks, various random graph models have been proposed, e.g. Preferential Attachment, Chung-Lu, or Hyperbolic.

We look at the interplay between the power-law degree distribution and the run time of optimization techniques for well known combinatorial problems. We observe that on scale-free networks, simple evolutionary algorithms (EAs) quickly reach a constant-factor approximation ratio on common covering problems

We prove that the single-objective \ea reaches a constant-factor approximation ratio on  the Minimum Dominating Set problem, the Minimum Vertex Cover problem, the Minimum Connected Dominating Set problem, and the Maximum Independent Set problem in expected polynomial number of calls to the fitness function.

Furthermore, we prove that the multi-objective \gsemo algorithm reaches a better approximation ratio than the \ea on those problems, within polynomial fitness evaluations.

\keywords{Evolutionary algorithms, covering problems, power-law bounded networks.}
\end{abstract}
%
%
\section{Introduction}
Bio-inspired randomized search heuristics, such as evolutionary algorithms, are well-suited to approach combinatorial optimization problems. These algorithms have been extensively analyzed on artificial pseudo-Boolean functions (see e.g. Droste et al. \cite{DJW02}, and Jansen and Wegener \cite{JW01}) as well as on some combinatorial optimization problems (see e.g. Giel and Wegener \cite{OW03}, Neumann \cite{Neumann07}, and Neumann and Wegener \cite{NW07}). The standard approach to perform theoretical analyses is the average-case black-box complexity: The run time is estimated by counting the expected number of calls to the valuation oracle or \emph{fitness function}. In recent years, a significant effort has been made to study the run time of EAs on combinatorial optimization problems.

The performance of single-objective population-based EAs on the Minimum Vertex Cover problem is studied in Oliveto et al. \cite{OHY08,OHY07}. They show that these algorithms achieve arbitrarily bad approximation guarantee on this problem. In contrast, Friedrich et al.~\cite{FHN07} prove that multi-objective EAs can obtain an optimal solution in expected polynomial time on this instance.

B\"{a}ck et al. \cite{BK94} study the performance of a simple a single-objective evolutionary algorithm for the Maximum Independent set problem and claim its superiority by experimental observations. Similarly, Peng \cite{Peng15} analyzes a single objective EA on this problem and proves that it reaches a $\left ( (\Delta(G)+1)/2 \right )$-approximation within expected $\bigo{n^4}$ fitness evaluations, where $n$ is the number of nodes in a given graph $G$, and $\Delta(G)$ its maximum node degree. These analyses, as well as those for the Minimum Vertex Cover problem, do not require any assumptions on the topology of $G$.

\paragraph{Power-law bounded networks.} A wide range of real-world networks, such as the internet graph, the web, power grids, protein-protein interaction graphs, and social networks exhibit properties such as high clustering coefficient (see Kumar et al. \cite{KumarRRT99}), small diameter (see Leskovech et al. \cite{LKF05}), and approximate power-law degree distribution (see Faolutsos et al. \cite{FaloutsosFF99}). Various models have been proposed to capture, and possibly explain, these properties (see e.g. Newman \cite{Newman03} and Watts and Strogatz \cite{WS98}).

One of the most significant contributions in this sense is that of Brach et al. \cite{PLBNetworks}. Many real-world networks approximately exhibit a power-law degree distribution. That is, the number of nodes of degree $k$ is approximately proportional to $k^{-\beta}$, for $k$ sufficiently large, where $\beta>1$ is a constant inherent to the network. Networks with this property are commonly referred to as \emph{scale-free}. Brach et al. \cite{PLBNetworks} define a deterministic condition that captures the behavior of the degree distribution of such networks. This condition requires an upper-bound on the number of nodes with degree $\delta(v) \in \left[2^i, 2^{i+1}\right)$, for all $i = 0,\cdots, \lfloor \log (n)-1\rfloor+1$. This property is commonly referred to as \PLBU. Carrying out run time analysis of EAs on networks that exhibit this property may give a better understanding of their performance in practice. 

\paragraph{Our contribution.} We study various well-known \NP-hard graph covering problems, and prove that simple EAs perform well, if the underlying graph exhibits the \PLBU property. Specifically, we study the Minimum Dominating Set problem,  the Minimum Vertex Cover problem, the Maximum Independent Set problem, and the minimum Connected Dominating Set problem.

We prove that for these problems a single-objective EA reaches a constant-factor approximation ratio within expected polynomial fitness evaluations. We prove that  a multi-objective EA reaches a better approximation guarantee, given more time budget.

The paper is organized as follows. In Section \ref{sec:preliminaries} we define the algorithms and give some basic definitions and technical tools. In Section \ref{sec:ds} we study the Minimum Dominating Set problem, both in the single- and multi-objective case. The results for the Minimum Vertex Cover problem follows directly from the analysis for the Minimum Dominating Set problem and are presented in Section \ref{sec:ds}. The Minimum Connected Dominating Set problem is discussed in Section \ref{sec:MCDSs-objsec}. We conclude with the analysis for the Maximum Independent Set problem, in Section \ref{sec:mis}.
%
%
%
%
%
%
%
\section{Preliminaries}
\label{sec:preliminaries}

In this paper, we only study undirected graphs $G=(V, E)$ without loops, where $V$ is  the set of nodes and $E$ is the set of edges. We always use $n$ to denote $\absl{V}$ and $m$ to denote $\absl{E}$.  We use $\degr(v)$ to denote the degree of each node $v\in V$, and we let $\Delta(G)$ be the maximum degree of $G$. Moreover, given a pseudo-Boolean array $x$ we denote with $\absl{x}_1$ the number of ones in the input string. Otherwise, we use the well-known mathematical and graph theoretic notations, e.g., we use $\absl{A}$ to denote the size of any set $A$. We use the following definition of \app.
\begin{definition}
Consider a problem $\mathcal{P}$ on a graph $G=(V, E)$ let $U$ be a possible solution, and denote with $\opt$ an optimal solution of $\mathcal{P}$. If $\mathcal{P}$ is a minimization problem, we say that $U$ is an \app if it holds that $\absl{U}/\absl{\opt} \leq \epsilon$. If $\mathcal{P}$ is a maximization problem, we say that $U$ is an \app if it holds that $\absl{\opt}/\absl{U} \leq \epsilon$.
\end{definition}
We consider both single- and multi-objective optimization. Whereas in the single-objective case the algorithm searches for the global optimum, in the latter the algorithm searches for a set of optimal solutions. Such solutions are part of the \emph{Pareto front}. 

Consider any two points $x' =(x'_1, \dots, x'_m)$ and $x'' =(x''_1, \dots , x''_m)$ in the $m$-dimensional space $\mathbb{R}^m$. We say that $x'$ \emph{dominates} $x''$, in symbols $x' \succ x''$, if it holds that $x'_i \geq x''_i$ for  all $i = 1, \dots, m$. With this notion of dominance, we define the Pareto front as follows.
\begin{definition}
\label{def:pareto_optimum}
Consider a function $f: X \subseteq \mathbb{R}^n \to \mathbb{R}^m$, with $X$ being a compact set in the metric space $\mathbb{R}^n$. Consider the set $Y:= \{  y \in \mathbb{R}^m \mid y = f(x), x \in X\}$, and denote with $\succ$ the standard partial order on $Y$. The Pareto front is defined as $\mathcal{P}(Y) = \{ y' \in Y \mid \{ y'' \succ y', y'' \neq y' \} = \emptyset \}$.
\end{definition}

\subsection{Algorithms}

The \ea is a randomized local search heuristic inspired by the process of biological evolution (cf. Algorithm \ref{1+1}). 
\begin{algorithm}[t]
	\caption{The \ea.}
    \label{1+1}
    \textbf{input:} a fitness function $f:2^V \rightarrow \mathbb{R}_{\geq 0}$\;
    \textbf{output:} an (approximate) global maximum of the function $f$\;
    $\qquad $\\
    // sample initial solution\\
 	choose $x\in \{0, 1\}^n$ uniformly at random\;
	$\qquad $\\
 	\While{convergence criterion not met}{
        $\qquad $\\
    	// perform mutation\\
   		create offspring $y$ by flipping each bit of $x$ independently w.p. $1/n$\;
        $\qquad $\\
    	// perform selection\\
		\If{$f(y) \leq f(x)$}{
			$x \gets y$\;
		}
        $\qquad $\\
   	}
    \textbf{return} x\;
\end{algorithm}
Initially, an \emph{individual} $x$ is sampled uniformly at random (u.a.r.). An offspring $y$ is then generated, by flipping all bits of $x$ independently with probability $1/n$. The fitness values of $x$ and $y$ are then compared. If the value $f(y)$ is less than or equal to $f(x)$, then $x$ is discarded and $y$ is preserved in memory; otherwise $y$ is discarded and $x$ is preserved in memory. Note that this algorithms are defined for minimization problems. However, symmetric definitions hold for maximization problems.\\
The \gsemo is a multi-objective evolutionary algorithm (cf. Algorithm \ref{GSEMO}). 
\begin{algorithm}[t]
\caption{The \gsemo.}
\label{GSEMO}
    \textbf{input:} a fitness function $f:2^V \rightarrow \mathbb{R}_{\geq 0}$\;
    \textbf{output:} an (approximate) global minimum of the function $f$\;
    $\qquad $\\
    // sample initial solution\\
    choose  $x\in \{0,1\}^n$ uniformly at random\;
    $P \gets P \cup \{x\}$\;
    $\qquad $\\    
\While{convergence criterion not met}{
    $\qquad $\\
    // perform mutation\\
	choose $x\in P$ uniformly at random\;
	create $y$ by flipping each bit of $x$ independently w.p. $1/n$\;
	$\qquad $\\
    // perform selection\\
	\If {$f(y)$ is not dominated by $f(z)$, for all points $z \in P$}{
		$P \gets P \cup \{y\}$\;
		delete all solutions $z\in P$ s.t. $f(z)$ is dominated by $f(y)$\;
	}
}
\Return $P$\;
\end{algorithm}
As in the case of the \ea, an initial solution is chosen u.a.r. from the objective space and stored it in the Pareto front $P$.  An element $x$ is then chosen u.a.r. from $P$ and  a new solution $y$ is generated by flipping each bit of $x$ independently w.p. $1/n$. If $y$ is not strongly dominated by any other solution in $P$, then $y$ is stored in the Pareto front, and all elements which are strongly dominated by $y$ are discarded. Otherwise, the population remains unchanged.

\subsection{The Multiplicative Drift theorem.}\label{sec:multiplicative}
In the case of the \ea, for any fitness function $f: \{0, 1\}^n \rightarrow \mathbb{R}_{\geq 0}$, we describe its run time as a Markov chain $\{X_t\}_{t\geq 0}$, where $X_t$ is the $f$-value reached at time step $t$. Following this convention, we perform parts of the analyses with the \emph{Multiplicative Drift Theorem} (see Doerr et al.~\cite{DBLP:journals/algorithmica/DoerrJW12}), a strong tool to analyze the run time of evolutionary algorithms such as the \ea . Intuitively, this theorem gives an estimate for the expected value of the run time of the \ea, provided that the change of the average value of the process $\{X_t\}_{t \geq 0}$ is within a multiplicative factor of the previous solution. More formally, the following theorem holds.
\begin{theorem}[Multiplicative Drift]
\label{thm:multiplicative_drift}
Let $S\subseteq \mathbb{R}$ be a finite set of positive numbers with minimum $s_{\min}$. Let $\{X^{(t)}\}_{t\in \mathbb{N}}$ be a random variables over $S \cup \{0\}$. Let $T$ be the random variable that denotes the first point in time $t \in \mathbb{N}$ for which $X^{(t)} = 0$.

Suppose that there exists a real number $\delta > 0$ such that
\[
\mathbb{E}[X^{(t)} - X^{(t + 1)} \mid X^{(t)} = s] \geq \delta s
\]
holds for all $s \in S$ with $Pr [X^{(t)} = s] > 0$. Then for all $s_0 \in S$ with $Pr [X^{(0)} = s_0] > 0$, we have
\[
\mathbb{E}[T\mid X^{(0)}  = s_0] \leq \frac{1 + \ln (s_0/s_{\min})}{\delta}.
\]
\end{theorem}
A proof of this result is given in Doerr et al. \cite[Theorem 3]{DBLP:journals/algorithmica/DoerrJW12}.
\section{Power-Law Bounded Networks}
\label{sec:tecDef}
In many real-world networks the degree distribution approximately follows a power-law. We frame this concept with the following definitions (see Brach et al. \cite{PLBNetworks}).
\begin{definition}[\PLBU Network]
\label{def:plbu}
An undirected graph  $G=(V, E)$ is a power-law bounded network, if there exists parameters $1<\beta=\mathcal{O}(1)$ and $t\ge 0$ s.t. 
the number of nodes $v$ with $\delta (v) \in [2^{d}, 2^{d + 1})$ is at most
\[
c_1 n(t + 1)^{\beta - 1} \sum_{i=2^d}^{2^{d+1}-1} (i + t)^{-\beta},
\]
for all $d\geq 0$, and with $c_1>0$ a universal constant.
\end{definition}
Definition \ref{def:plbu} intuitively captures the general idea that vertices with bounded degree can be grouped into sets with cardinality upper-bounded by a power-law. Again, we remark that this property is \emph{deterministic}. We say a graph is a \emph{\PLBU networks} if it is power-law bounded as in Definition \ref{def:plbu}. 

We prove that \PLBU networks have the property that feasible solutions of various covering problems yield a constant-factor approximation ratio. To this end, we consider the following definition.
\begin{definition}
\label{def:dominating_set}
Consider an undirected graph $G=(V, E)$. We say that a set $D\subseteq V$ is \emph{dominating} if every node of $V\setminus D$ is adjacent to at least one node of $D$.
\end{definition}
Dominating sets play an important role in our analysis, because the sum of the degree of their nodes always yields constant-factor approximation with respect (w.r.t.) to their cardinality, on \PLBU networks with bounded parameters $\beta, t$. More formally, the following lemma holds.
\begin{lemma}
\label{lemma:potential_volume}
Let $G=(V, E)$ be a \PLBU network with parameters $2< \beta = \bigo{1}$, $t \geq 0$, and a universal constant $c_1$. Denote with $\optd$ a dominating set of $G$. Define the constants \abeq Then it holds
\[
\frac{\sum_{v \in \optd}(\delta(v) + 1)}{\absl{\optd}} \leq 2 a b + 1 = \Theta \left ( 1 \right ).
\]
\end{lemma}
\begin{proof}
Define the following constants
\begin{equation*}
d_{\geq k} := \absl{\{v \in V \colon \delta(v) \geq k\}} \quad \mbox{and}\quad\gamma := \min \{k \in \mathbb{N}\colon d_{\geq k} \leq \absl{\optd}\}.
\end{equation*}
Then $d_{\geq k}$ is the number of nodes of $G$ with degree at least $k$, and $\gamma$ is the smallest index $k$ by which the nodes of $G$ with degree at least $2^k$ is at most $\absl{\optd}$. We first prove an upper-bound on the numerator, in terms of $\gamma$. From Definition \ref{def:dominating_set} it holds
\begin{align}
\sum_{v \in \optd}\delta(v) & \leq c_1 n(t + 1)^{\beta - 1} \sum_{j = \gamma}^{\lceil \log (n - 1) \rceil} 2^{j + 1} \sum_{i = 2^j}^{2^{j + 1} - 1}(i + t)^{-\beta}\nonumber \\
& \leq 2 c_1 n(t + 1)^{\beta - 1} \sum_{i = 2^{\gamma}}^{2^{\lceil \log (n - 1) \rceil + 1}- 1} (i + t)^{1-\beta}\nonumber \\
& \leq 2 c_1 n(t + 1)^{\beta - 1} \int_{2^{\gamma}}^{2^{\lceil \log (n - 1) \rceil + 1}- 1} (x + t)^{1-\beta}dx\nonumber \\
& \leq 2 c_1 n(t + 1)^{\beta - 1} \frac{(2^{\gamma} + t)^{2 - \beta}}{\beta -2}.\label{eq:4potential_vol}
\end{align}
We continue by computing a lower-bound on $\absl{\optd}$ in terms of $\gamma$. Again, from Definition \ref{def:dominating_set} we have that it holds
\begin{align}
\absl{\optd} \geq n_{\gamma} & = c_1 n (t +1)^{\beta - 1} \sum_{i = 2^{\gamma}}^{2^{\lceil \log (n - 1) \rceil + 1}- 1} (i + t)^{-\beta}\nonumber \\
& \geq c_1 n (t +1)^{\beta - 1} \int_{2^{\gamma}}^{2^{\gamma + 1}} (x + t)^{-\beta}dx\nonumber \\
& \geq c_1 n (t +1)^{\beta - 1} \frac{1 - \left (\frac{t + 2}{t + 1} \right )^{1 - \beta}}{\beta - 1} (2^\gamma + t)^{1-\beta}.\label{eq:5potential_vol}
\end{align}
We combine the upper-bound on the numerator $\sum_{v \in \optd}\delta(v) + 1$ given in \eqref{eq:4potential_vol}, together with the lower bound on $\absl{\optd}$ as in \eqref{eq:5potential_vol}, to obtain that it holds
\begin{equation}
\frac{\sum_{v \in \optd} (\delta(v) + 1)}{\absl{\optd}} \leq 2 \frac{\beta - 1}{\beta - 2} \left [1 - \left (\frac{t + 2}{t + 1}\right )^{1 - \beta} \right ]^{-1}(2^\gamma + t) + 1.\label{eq:first_estimate_ds}
\end{equation}
We conclude by giving an upper-bound for $2^\gamma + t$. From the definition of dominating set (see Definition \ref{def:dominating_set}) it holds $\sum_{v \in \optd}\delta(v) \geq n/2$. Combining this observation with \eqref{eq:4potential_vol}, it holds
\[
\frac{n}{2} \leq 2 c_1 n(t + 1)^{\beta - 1} \frac{(2^{\gamma} + t)^{2 - \beta}}{\beta -2},
\]
from which it follows that
\begin{equation}
2^{\gamma} + t \leq \left [ 4 c_1 \frac{(t + 1)^{\beta - 1}}{\beta - 2} \right ]^{\frac{1}{\beta - 2}}.
\label{eq:second_estimate_ds}
\end{equation}
The claim follows by substituting \eqref{eq:second_estimate_ds} in \eqref{eq:first_estimate_ds}.
\end{proof}
Note that Lemma \ref{lemma:potential_volume}  uses a parameter $\beta > 2$. From a practical point of view this assumption is not restrictive, since most real-world networks fulfill this requirement. However, for $\beta < 2$ one may define a degenerate \PLBU network for which our analysis fails.

Lemma \ref{lemma:potential_volume} can be readily used used to prove approximation guarantees for the Minimum Dominating Set problem - i.e. the problem of searching for a minimum dominating set. In the following, however, we show that this lemma can be used to derive approximation guarantees for other common covering problems.

Given a graph $G=(V, E)$, a \emph{connected dominating set} is any dominating set $C \subseteq V$ s.t. the nodes of $C$ induce a connected sub-graph of $G$. From Lemma \ref{lemma:potential_volume} the following result readily follows.
\begin{corollary}
\label{cor:cds}
Let $G=(V, E)$ be a \PLBU network with parameters $2< \beta = \bigo{1}$, $t \geq 0$, and a universal constant $c_1$. Denote with $\optc$ a connected dominating set of $G$. Define the constants \abeq Then it holds
\[
\frac{\sum_{v \in \optc}(\delta(v) + 1)}{\absl{\optc}} \leq 2 a b + 1 = \Theta \left ( 1 \right ).
\]
\end{corollary}

Finally, we apply Lemma \ref{lemma:potential_volume} to independent sets. Given a graph $G=(V, E)$, an \emph{independent set} $M\subseteq V$ consists of vertices of $G$, no two of which are adjacent. A \emph{maximum independent set} is any independent set of maximum size. It is not true, in general, that an independent set is a dominating set. However, a maximum independent set is always a dominating set. Hence, the following corollary holds.
\begin{corollary}
\label{cor:mis}
Let $G=(V, E)$ be a \PLBU network with parameters $2< \beta = \bigo{1}$, $t \geq 0$, and a universal constant $c_1$. Denote with $\optm$ a maximum independent set of $G$. Define the constants \abeq Then it holds
\[
\frac{\sum_{v \in \optm}(\delta(v) + 1)}{\absl{\optm}} \leq 2 a b + 1 = \Theta \left ( 1 \right ).
\]
\end{corollary}
%
%
%
%
\section{The Minimum Dominating Set Problem}\label{sec:ds}
We study the minimum dominating set problem (\ds): For a given graph $G=(V,E)$, find a dominating set of $V$ (see Definition \ref{def:dominating_set}) of minimum cardinality. In this section, we denote with $\opt$ any optimal solution to the \ds. Given a graph $G=(V, E)$ with $|V|=n$, we consider an indexing $i:V\rightarrow \{1, \dots, n \}$ of the vertices. We represent any set of nodes $S \subseteq V$ with a pseudo-Boolean array $(x_1, \dots, x_n)$ of length $n$ and s.t. $x_i = 1 $ if and only if the $i$-th node of $G$ is in $S$.
%
%
\subsection{Single-objective optimization.}
\label{subsec:single_obj_mds}
We search for a minimum dominating set by minimizing the following single-objective function 
\begin{equation}
\label{eq:fitness_ds}
F(x)= n u(x)+\lvert x\rvert _1,
\end{equation}
where $u(x)$ is the number of non-dominated nodes, and $\lvert x\rvert _1$ is the number of $1$s in the input string. Note that $u(x) = 0$ if and only if the solution $x$ is a dominating set. We prove the main result of this section, by giving an upper-bound on the run time of the \ea until a dominating set is found. We then use Lemma \ref{lemma:potential_volume} to give an upper-bound on the approximation ratio achieved by this solution.    

\begin{theorem}
\label{Thm1}
Let $G=(V, E)$ be a \PLBU network with parameters $\beta>2$, $t\geq 0$ and with a universal constant $c_1$.  Define the constants \abeq  Then the \ea finds a $(2ab + 1)$-approximation of the \ds after expected $\bigo{n \log n}$  fitness evaluations. In particular, the \ea is a $\Theta(1)$-approximation algorithm for the \ds on \PLBU networks.
 \end{theorem}
\begin{proof}
We first observe that because of the weights on the fitness, the algorithm initially searches for a dominating set and then tries to minimize it. Furthermore, no step that reduces the number of dominated nodes is accepted. We estimate the expected run time until a feasible solution is found, and conclude by proving that any feasible solution reaches the desired approximation guarantee.

To estimate the run time we use the Multiplicative Drift theorem (see Theorem \ref{thm:multiplicative_drift}),  by defining $\{X^{(t)} \}_{t \geq 0}$ as the process $X^{(t)} = u(x_t)$ for all $t \geq 0$. Suppose that the current solution $x_t$ is not a dominating set - i.e. $u(x_t) > 0$. Then there exists a node $v \in V$ s.t. by adding it to the current solution, the fitness $u(x_t)$ decreases by $1$. Since the probability of performing a single chosen bit-flip is at least $1/en$, then it holds
\[
\mathbb{E}[X_t - X_{t - 1} \mid X_t] \geq \frac{X_t}{e n}.
\]
Following the notation of Theorem~\ref{thm:multiplicative_drift}, we can trivially set $S = \{1, \dots, n\}$ and obtain that the run time until a feasible solution is found is upper-bounded as $\bigo{n \log n}$. We conclude by proving that the \ea reaches the desired approximation guarantee. Denote with $x^*$ the first dominating set reached by the \ea. We observe that any optimal solution \opt dominates all nodes of $G$. It follows that it holds
\begin{equation}
\label{eq:last_expected_decrease}
\frac{\absl{x^*}}{\absl{\opt}} \leq \frac{\sum_{v \in \opt} (\delta(v) + 1)}{\absl{\opt}} \leq 2ab + 1,
\end{equation}
where the last inequality in \ref{eq:last_expected_decrease} follows from Lemma \ref{lemma:potential_volume}. The claim follows.
\end{proof}
%
%
\subsection{Multi-objective optimization.}
We approach \ds with the multi-objective $\gsemo$ (cf. Algorithm \ref{GSEMO}). In this case, we use the following bi-objective fitness function
\begin{equation}
f=(u(x),\lvert x\rvert _1),
\label{eq:fintess_ds_gsemo}
\end{equation}
where $u(x)$ is the number of non-dominated nodes, and $\lvert x\rvert _1$ the number of $1$s in the input string. We prove that the \gsemo fins a good approximation of the \ds in polynomial time, when optimizing a fitness as in \eqref{eq:fintess_ds_gsemo}. Useful in the analysis is the following lemma.
\begin{lemma}
\label{lemDS1}
Let $G=(V, E)$ be a graph, and let $S_k := \{v_1, \dots, v_k\}$ be a sequence of nodes s.t. $v_j$ has the maximum degree in the complete sub-graph induced by the nodes of $G$ that are not dominated by $\{v_1, \dots, v_{j - 1}\}$. Let $n_k$ be the number of nodes in $V$ that are not dominated by $S_k$. Then it holds
\[
n_{k}\leq n\left(1-\frac{1}{\lvert \opt \rvert }\right)^{k},
\]
for all $0 < k \leq n$.
\end{lemma}
\begin{proof}
We prove this by induction on $k$. For the base case, we set $k=1$. Since the number of non-dominated nodes is $n$, there exists a node $v\in V$ s.t. $v$ dominates at least $n/\lvert \opt \rvert$ nodes in the graph. Otherwise,  $\opt $ is not an optimal solution. Thus, the number of non-dominated nodes by $S_1$ is at most
\[
n_1\leq n-\frac{n}{\lvert \opt \rvert } = n \left(1-\frac{1}{\lvert \opt \rvert }\right).
\]
For the inductive step, suppose that the statement holds for $k=i$. Again, since $|\opt|$ is the size of the minimum dominating set, by the pigeon-hole principle there exists a node $ V\setminus S_i$ s.t. $S_i \cup \{v\}$ dominates at least $n_i/\absl{\opt}$ many nodes. Therefore, it holds that 
\[
n_{i+1} \leq n_i - \frac{n_i}{\absl{\opt }} = n_i \left (1 - \frac{1}{\absl{\opt}} \right ) = n  \left (1 - \frac{1}{\absl{\opt}} \right ) ^{i+1}.
\]
The claim follows.
\end{proof}
We use the lemma above to analyze the run time and approximation guarantee of  the \gsemo on \ds. The following theorem holds.
\begin{theorem}
\label{GSEMO_DS}
Let $G=(V, E)$ be a \PLBU network with parameters $\beta>2$, $t\geq 0$ and with a universal constant $c_1$.  Define the constants \abeq  Then  the \gsemo finds an $\ln(2ab+1)$-approximation for the \ds after expected $\bigo{n^3}$ fitness evaluations. In particular, \gsemo is a $\Theta(1)$-approximation algorithm for the \ds on \PLBU networks.
\end{theorem}
\begin{proof}
The proof consists of two parts. We first estimate the expected run time until the solution $0^n$ - i.e. the empty set - is added to the Pareto front (Phase 1), and then we estimate the time until the desired approximation is reached from that solution (Phase 2).

(Phase 1) To analyze the first part of the process, we denote with $x^*$ the individual in the Pareto front with the minimum number of $1$s, at a given time step. Since the total number of elements in the Pareto front is upper-bounded by $n+1$, the probability of choosing $x^*$ for mutation is at least $1/(n+1)$ at each iteration. Once $x^*$ is selected for mutation, then the probability of performing a single bit-flip and remove one element to $x^*$ is at least $\absl{x^*}_1/(en)$. Therefore, the expected waiting time to reach the solution $0^n$ is upper-bounded as
\[
(n+1) \sum_{j=1}^{n} \frac{e n}{j} = \bigo{n^2 \log n}.
\]

(Phase 2) We now give an upper bound on the run time until the desired approximation is reached, assuming that the solution $0^n$ has been found. From Lemma \ref{lemDS1}, it follows that there exists a node $v_1\in V$ s.t. after adding this node to $0^n$ the number of non-dominated nodes is at most
\[
n_1 := u\left (\{ v_1\} \right ) \leq n \left ( 1 - \frac{1}{\absl{\opt}} \right ).
\]
With an argument similar to the one presented to analyze the first part of the process, in expected $\bigo {n^2}$ time an individual with fitness value $(n_1,1)$  can be added to the solution. If $\{ v_1\}$ is not an optimal solution, then by Lemma \ref{lemDS1} there exists a node $v_2 \in V \setminus \{v_1\}$ s.t.
 \[
n_2 := u\left (\{ v_1, v_2\} \right ) \leq n \left ( 1 - \frac{1}{\absl{\opt}} \right )^2.
\]
Hence, an individual $(n_2, 2)$ can be added to the current population within expected $\bigo{n^2}$, after a solution with fitness $(n_1,1)$ was added. Similarly, an individual with fitness $(n_i, i)$ can be added to the current population within expected $\bigo{n^2}$, after a solution with fitness $(n_{i-1}, i-1)$ was added. We can iterate this process as indicated in Lemma \ref{lemDS1}, until a dominating set $S_k = \{v_1, \dots , v_k\} \subseteq V$ is found.
\begin{equation}
\label{eq:GSEMOar1}
n_k := u\left (S_k \right ) \leq n \left ( 1 - \frac{1}{\absl{\opt}} \right )^k.
\end{equation}
Note that there are at most $n$ nodes that can be added to the solution $0^n$ as described above. Hence after a total of $\bigo{n^3}$ fitness evaluations the \gsemo reaches a solution $S_k$ as in Lemma \ref{lemDS1} that is a dominating set.

We conclude by showing that the solution $S_k$ yields the desired approximation guarantee. Again, from the definition of minimum dominating set it holds $n \leq \sum_{v \in \opt} \delta (v) + 1$. Combining this observation with \eqref{eq:GSEMOar1}, we can estimate $k$ by solving the inequality
\begin{equation}
\label{eq:MDS}
\frac{\absl{S_k}}{\absl{\opt}} \leq \frac{n}{\absl{\opt}} \left ( 1 - \frac{1}{\absl{\opt}} \right )^{k} \leq  \frac{\sum_{v \in \opt}(\delta (v) + 1)}{\absl{\opt}}\left ( 1 - \frac{1}{\absl{\opt}} \right )^{k}.
\end{equation}
We combine \eqref{eq:MDS} with Lemma \ref{lemma:potential_volume} to obtain that
\[
1 \leq (2ab + 1) \exp \left( - \frac{k}{\absl{\opt}} \right).
\]
By taking the logarithm on both sides, we get $k \leq \absl{\opt}\ln(2ab + 1)$ and we can conclude that the cardinality of $S_k$ is at least $\absl{S_k} \leq \absl{\opt}\ln(2ab + 1)$. The claim follows.
\end{proof}
%
%
%
%
\section{The Minimum Vertex Cover Problem}
Given a graph $G=(V,E)$, a vertex cover is a set $S\subseteq V$ of vertices s.t. each edge in $E$ is adjacent to at least one node in $S$. The minimum vertex cover problem (\vc) consists of finding a vertex cover of minimum size. In this section, we denote with \opt any optimal solution to the \vc, and we use intuitive bit-string representation based on vertices, as discussed in Section \ref{sec:ds}.
%
%
\subsection{Single-objective optimization.}
We study the run time and approximation guarantee for the \ea. Since any vertex cover is also a dominating set, then one can approach the \vc by minimizing a fitness function as in \eqref{eq:fitness_ds}. Moreover, the following run time analysis follows directly from Theorem \ref{Thm1}.
\begin{corollary}
\label{EA_VC}
Let $G=(V, E)$ be a \PLBU network with parameters $\beta>2$, $t\geq 0$ and with a universal constant $c_1$.  Define the constants \abeq  Then the \ea finds a $(2ab)$-approximation of the \vc in expected $\bigo{n \log n}$  fitness evaluations. In particular, the \ea is a $\Theta(1)$-approximation algorithm for the \vc on \PLBU networks.
\end{corollary}
%
%
\subsection{Multi-objective optimization.}
We prove that the multi-objective \gsemo yields an improved approximation guarantee then the \ea, at the cost of a higher number of expected fitness evaluations. Again, since a vertex cover is also a dominating set, then this problem can be approached by minimizing a two-objectives function as in \eqref{eq:fintess_ds_gsemo}. Hence, the following run time analysis for the \gsemo readily follows from Theorem \ref{GSEMO_DS}.
\begin{theorem}
\label{GSEMO_VC}
Let $G=(V, E)$ be a \PLBU network with parameters $\beta>2$, $t\geq 0$ and with a universal constant $c_1$.  Define the constants \abeq  Then  the \gsemo finds an $(\ln(2ab)+1)$-approximation solution of \vc in expected $\bigo{n^3}$ fitness evaluations. In particular, \gsemo is a $\Theta(1)$-approximation algorithm for the \vc on \PLBU networks.
\end{theorem}
%
%
%
%
\section{The Minimum Connected Dominating Set Problem}
\label{sec:MCDSs-objsec}
Given a graph $G=(V,E)$, recall that a connected dominating set is a set $D\subseteq V$ of vertices that dominates every node $v\in V$, and the nodes $u\in D$ induce a connected sub-graph on $G$. The minimum connected dominating set problem (\cds) consists of finding a connected dominating set $C\subseteq V$ of minimum size. In this section, we denote with \opt any optimal solution to the \cds. Throughout this section, we use intuitive bit-string representation based on vertices, as discussed in Section \ref{sec:ds}.
%
%
\subsection{Single-objective optimization.}
\label{sec:MCDSs-obj}
We study the optimization process for the \cds with the \ea. We approach this problem by minimizing following fitness function.
\begin{equation}
\label{eq:fitness_single_objective cds}
F(x)=n^2(u(x)+ w(x) - 1)+\rvert x\lvert _1,
\end{equation}
where $u(x)$ is the number of non-dominated nodes, and $w(x)$ is the number of connected components in the sub-graph induced by the chosen solution. Note that any connected dominating set $C$ yields $u(C) + w(C) - 1 = 0$. We search for a solution to the \cds by minimzing $F(x)$. It can be observed that a sub-graph induced by any subset $D\subseteq V$ that is represented by the input string $x$ is a connected dominating set if and only if $u(x)=0$ and $w(x)=1$. Because of the weights on the fitness function, the \ea tends to reach a feasible solution first, and then it removes unnecessary nodes.

The argument for the theorem below is similar to that given in Theorem \ref{Thm1}:  We first show that the \ea reaches a locally optimal solution in $\bigo{n \log n}$ fitness evaluations, and then we use the \PLBU property to show that any such solution gives the desired approximation ratio (see Corollary \ref{cor:cds}). The following theorem holds.
\begin{theorem}\label{CDSEA}
Let $G=(V, E)$ be a \PLBU network with parameters $\beta>2$, $t\geq 0$ and with a universal constant $c_1$.  Define the constants \abeq  Then the \ea finds a $(2ab)$-approximation of the \cds in expected $\bigo{n \log n}$  fitness evaluations. In particular, the \ea is a $\Theta(1)$-approximation algorithm for the \cds on \PLBU networks.
\end{theorem}
\begin{proof}
To estimate the run time we use the multiplicative Drift theorem (see Theorem \ref{thm:multiplicative_drift}). Define the function $f(x) = u(x) + w(x) - 1$, with $u(x)$ and $w(x)$ as in \eqref{eq:fitness_single_objective cds}, and let $x_t$ be a solution found at time step $t$. Suppose that the current solution $x_t$ is not feasible. Then with probability atl east $1/en$ a node is added to $x_t$ s.t. the value $f(x_t)$ decreases by $1$. Hence, if we define $\{X_t \}_{t \geq 0}$ as the process $X_t = u(x_t)$ for all $t \geq 0$, then it holds
\[
\mathbb{E}[X_t - X_{t - 1} \mid X_t] \geq \frac{X_t}{e n}.
\]
Following the notation of Theorem~\ref{thm:multiplicative_drift}, we set $S = \{ 1, \dots, n \}$ to obtain that the run time until a feasible solution is found is upper-bounded as $\bigo{n \log n}$.

We conclude by proving that the \ea reaches the desired approximation guarantee. Denote with $x^*$ the first feasible solution reached by the \ea. We observe that any optimal solution \opt dominates all nodes of $G$. It follows that it holds
\begin{equation}
\label{eq:last_expected_decrease2}
\frac{\absl{x^*}}{\absl{\opt}} \leq \frac{\sum_{v \in \opt} (\delta(v) + 1)}{\absl{\opt}} \leq 2ab + 1,
\end{equation}
with the last inequality in \eqref{eq:last_expected_decrease2} following from Corollary \ref{cor:cds}. The claim follows.
\end{proof}
%
%
\subsection{Multi-objective optimization.}
As in the case of the \ds and the \vc, we again analyze the run time of the \gsemo. We search for a solution of the \cds by minimizing the following multi-objective fitness
\begin{equation}
\label{eq:fitness_cds_multi_obj}
F(x) = (u(x) + w(x), \absl{x}_1),
\end{equation}
with $u(x)$ and $w(x)$ as described earlier in this section. The following result is useful for the run time analysis.
\begin{lemma}
\label{lemCDS1}
Let $G=(V, E)$ be a connected graph, and consider the function $f(x) := u(x) +w(x)$, with $u(x), w(x)$ as in \eqref{eq:fitness_cds_multi_obj}. Let $S_k := \{v_1, \dots, v_k\}$ be a set of nodes of $G$ s.t. 
\[
v_j := \argmax_{u \in V \setminus \{v_1, \dots, v_{j  - 1}\}} \left \{ f(S_k\setminus \{v_k\}) - f(S_k) \right \},
\]
for all $j = 1, \dots, k$. Then there exists a node $v \in V \setminus S_k$ s.t. it holds that
\[
f(S_k\cup \{v\}) \leq f(S_k)- \frac{f(S_k)}{\absl{\opt}} + 1.
\]
\end{lemma}
A proof of the lemma above is given by Guha and Khuller \cite[Theorem~3.3]{GK98}, where it is used to analyze the performance of a greedy algorithm. The following theorem holds.
\begin{theorem}\label{GSEMO_CDS}
Let $G=(V, E)$ be a connected \PLBU network with parameters $\beta>2$, $t\geq 0$ and with a universal constant $c_1$.  Define the constants \abeq  Then the \gsemo finds a $\ln (2 e a b + e)$-approximation for the \cds after expected $\bigo{n^3}$ fitness evaluations. In particular, the \gsemo is a $\Theta(1)$-approximation algorithm for the \cds on \PLBU networks.
\end{theorem}
\begin{proof}
The hereby presented argument is similar to the one given in Theorem~\ref{GSEMO_DS}. Since the function $\absl{x}_1$ is an objective, then there are at most $n+1$ individuals in the Pareto front, at each point during the iteration. Using this knowledge, and following an argument as in (Phase 1) in Theorem \ref{GSEMO_DS}, one can prove that the solution $0^n$ is added to the Pareto front after expected
\[
\sum_{j = 1}^{n}\left (\frac{j}{e(n+1)n}\right )^{-1} = \bigo{n^2 \log n}
\]
fitness evaluations. 

Once the empty set is reached, the \gsemo iteratively generates solutions with increasing cardinality, until a good approximation of a minimum dominating set is reached. Define $f(x) = u(x) + w(x)$, and denote with $S_j = \{v_1, \dots, v_j\}$ a sequence of nodes s.t.
\[
v_j := \argmax_{u \in V \setminus \{v_1, \dots, v_{j  - 1}\}} \left \{ f(S_k\setminus \{v_k\}) - f(S_k) \right \},
\]
until a connected dominated set $S_k$ is reached (see Lemma \ref{lemCDS1}). Again, selecting an individual for mutation and adding a chosen node to it occurs after expected $\bigo{n^2}$ fitness evaluations. Since there are at most $n$ nodes that can be added to the empty set $0^n$, then after expected $\bigo{n^3}$ fitness evaluations the \gsemo reaches a solution $S_k$ that is a dominating set.

We conclude by showing that $S_k$ yields the desired approximation guarantee. From Lemma \ref{lemCDS1} it holds
\begin{align*}
f(S_k) &\leq f(S_0)\left(1-\frac{1}{\lvert \opt \rvert }\right)^k+\sum_{j=0}^{k-1}\left(1-\frac{1}{\lvert \opt \rvert }\right)^j\\
& \leq f(S_0)\left(1-\frac{1}{\lvert \opt \rvert }\right)^k+\absl{\opt}\left(1-\left(1-\frac{1}{\lvert \opt \rvert }\right)^k \right)\\
&\leq \left ( f(S_0)-\lvert \opt \rvert \right )\left(1-\frac{1}{\lvert \opt \rvert }\right)^k +\lvert \opt \rvert \\
& = \left ( n+1-\lvert \opt \rvert \right )\left(1-\frac{1}{\lvert \opt \rvert }\right)^k +\lvert \opt \rvert.
\end{align*}
Note also that from Corollary \ref{cor:cds} it holds
\[
\frac{n + 1 -\absl{\opt}}{\absl{\opt}} \leq \frac{\sum_{x\in \opt} (\delta(x)+1)}{\absl{\opt}}\leq 2ab + 1,
\]
and the following chain of inequalities also holds
\begin{equation}
1 \leq \frac{f(S_k)}{\absl{\opt}} \leq (2ab + 1)\left(1-\frac{1}{\lvert \opt \rvert }\right)^{k} + \absl{\opt}. \label{eq:cds_gsemo_label}
\end{equation}
If we solve \eqref{eq:cds_gsemo_label} w.r.t. $k$, then it follows that $k \leq \absl{\opt} \ln (2 e a b + e)$. Since $k = \absl{S_k}$, then we conclude that $\absl{S_k}/\absl{\opt} \leq \ln (2 e a b + e)$, and the claim follows.
\end{proof}
%
%
%
%
\section{The Maximum Independent Set Problem} \label{sec:mis}
For a given graph $G=(V,E)$, recall that an independent set is a subset $S\subseteq{V}$ s.t. no two nodes $u,v\in S$ are adjacent. The maximum independent set problem (\is) consists of finding an independent set $S\subseteq V$ of maximum cardinality. In this section, we denote with \opt any solution for the \is. In this section, we use a standard bit-string representation based on vertices, as discussed in Section \ref{sec:ds}.
%
%
\subsection{Single-objective optimization.}
To implement the \ea, we use the fitness function proposed by B\"{a}ck et al.~\cite{BK94}, which is defined as follows:
\begin{equation}
F(x)=\lvert x\rvert _1 - n\sum_{i=1}^n x_i \sum_{j=1}^n x_j e_{ij},\label{eq:maximum_indipendent_set_fitness}
\end{equation}
where $e_{ij}$ is $1$ if there is an edge between $v_i$ and $v_j$, and it is $0$ otherwise. In this case, the objective of the \ea is to maximize $F(x)$. Note that it holds that~$ \sum_{i=1}^n x_i \sum_{j=1}^n x_j e_{ij} =0$ if and only if the solution is an independent set. Because of the weight $n$, the \ea reaches a feasible solution first and then it adds nodes to it, to maximize the fitness. 

In order to perform the analysis for the \ea, we use the following definition of local optimality.
\begin{definition}
\label{def:3_bit_optimality}
Let $G=(V, E)$ be a graph, and let $F$ be as in \eqref{eq:maximum_indipendent_set_fitness}. We say that a set of nodes $S \subseteq V$ is a $3$-local optimum of it holds
\[
F((S \setminus U) \cup T) \leq F(S)
\]
for all $U \subseteq S$ and $T \subseteq V \setminus S$ s.t. $\absl{U \cup T} \leq 3$.
\end{definition}
Definition \ref{def:3_bit_optimality} captures the notion of optimality up to three bit-flips. Using this definition, we prove the following result. 
\begin{theorem}\label{lem:LS}
Let $G=(V, E)$ be a connected \PLBU network with parameters $\beta>2$, $t\geq 0$ and with a universal constant $c_1$.  Define the constants \abeq  Then the \ea finds an $(ab + 1/2)$-approximation for the \is after expected $\bigo{n^4}$ fitness evaluations. In particular, the \ea is a $\Theta(1)$-approximation algorithm for the \is on \PLBU networks.
\end{theorem}
\begin{proof}
We first perform a run time analysis for the \ea, until a $3$-local optimum is found (see Definition \ref{def:3_bit_optimality}), and then we prove that any such solution yields the desired approximation guarantee.

To perform the run time analysis, we divide the optimization process in two phases. In (Phase 1) the \ea searches for a feasible solution, whereas in (Phase 2) it searches for a $3$-local optimum, given that a feasible solution has been found.\\
(Phase 1) Let $x_0$ be the initial solution of the \ea and assume that $x_0$ is not feasible. Since the empty set is always an independent set, then we can use an argument as in Theorem \ref{Thm1} (Phase 1) to conclude that the \ea reaches a feasible solution after $\bigo{n \log n}$ fitness evaluations.\\
(Phase 2) After a feasible solution is reached, then \ea searches for a $3$-local optimum, within the portion of the search space that consists of feasible solutions. In this phase the \ea maximizes the function $\absl{x}_1$. We call \emph{favorable} any move up to a $3$-bit flip that yields an improvement on $\absl{x}_1$. The probability that the \ea performs a favorable move is $1/(en)^3$. Hence, it is always possible to obtain an improvement on the fitness of at least $1$, after expected $\bigo{n^3}$ fitness evaluations, unless the current solution is a $3$-local optimum. Note that after at most a linear number of favorable moves the \ea reaches a $3$-local optimum. It follows that the \ea finds a $3$-local optimum as in Definition \ref{def:3_bit_optimality} after expected $\bigo{n^4}$ fitness evaluations.

To prove that $3$-local optima reach the desired approximation guarantee, we present an argument similar to that of Khanna et al. \cite{KMS98}. Let $S$ be a $3$-local optimum, and define $I:=S\cap \opt $. Since we cannot add a node outside of $S$ to get the bigger independent set (see Definition \ref{def:3_bit_optimality}), then every node in $S$ has at least one incoming edge from $V\setminus S$. Moreover, because of the local optimality of $S$, there are at most $\absl{S \setminus I}$ nodes in $\absl{ \opt \setminus I }$ that have exactly one edge coming from $S$. Thus, $\absl{ \opt } -\absl{S}$ nodes in $\opt  \setminus  I$ must have at least two outgoing edges to $S$. This implies that the minimum number of edges between $S$ and $\opt \setminus  I$ is  $\absl{ S} -\absl{ I} +2(\absl{ \opt } -\absl{ S} )$. If we denote with $E(S, \opt \setminus I)$ the set of all edges between $S$ and $\opt \setminus  I$, then it holds
\begin{equation}
\frac{2 \absl{\opt}}{\absl{S}} \leq \frac{ \absl{S} - \absl{I} +2(\absl{\opt} - \absl{S} )}{\absl{S}} \leq \frac{E(S, \opt \setminus I)}{\lvert S\rvert }\leq \frac{\sum_{v\in S}(\delta(v)+1)}{\absl{S}}. \label{eq5}
\end{equation}
From the definition of $3$-local optimum it follows that $S$ is a maximum independent set. Hence, we apply Corollary \ref{cor:mis} to \eqref{eq5} and conclude that $\absl{\opt}/\absl{S} \leq a b + 1/2 $, as claimed.
\end{proof}
%
%
\subsection{Multi-objective optimization.}
We also consider the multi-objective case for \is. As in the previous sections, we consider the bi-objective function $f = (\absl{x}_1, u(x))$, with $u(x)$ defined as 
\[
u(x) = - \sum_{i=1}^n x_i \sum_{j=1}^n x_j e_{ij},
\]
where $e_{ij}$ is $1$ if there is an edge between $v_i$ and $v_j$, and it is $0$ otherwise. The \gsemo attempts to maximize both objectives simultaneously. To perform the run time analysis of \gsemo, we use an argument similar to that of Halld\'{o}rsson and Radhakrishnan \cite{DBLP:journals/algorithmica/HalldorssonR97} for the greedy algorithm. 
\begin{theorem}
\label{thm:average_degree}
Let $G=(V, E)$ be a graph. For any node $v \in V$, denote with $N(v)$ the set of nodes adjacent to $v$. Let $S_k := \{v_1, \dots, v_k\}$ be a sequence of $k$ nodes s.t. $v_j$ is the node with smallest degree in the complete sub-graph induced by $V\setminus \left ( \bigcup_{i=1}^{j-1} (N(v_i)\cup v_i )\right )$, and with $\bigcup_{i=1}^{k-1} (N(v_i)\cup v_i ) = \emptyset $. Then $S_k$ is a maximum independent set. Furthermore, it holds
\[
\absl{S_k} \geq \absl{V} \left ( \frac{\sum_{v\in V}\delta(v)}{\absl{V}} + 1 \right )^{-1}.
\]
\end{theorem}
A proof of Theorem \ref{thm:average_degree} is given in Halld\'{o}rsson and Radhakrishnan \cite[Theorem 1]{DBLP:journals/algorithmica/HalldorssonR97}. 

We also consider a lemma that allows to obtain an upper-bound on the average degree of a \PLBU network $G$. The following lemma holds.
\begin{lemma}
\label{lemma:plb_sum}
Let $G=(V, E)$ be a \PLBU network with parameters $\beta > 2, t \geq 0$, and with universal constant $c_1$. Then it holds
\[
\sum_{v \in V} \delta (v) \leq 2 c_1 n (t + 1)^{\beta-1} \sum_{i = 1}^{\Delta(G)} i(i + t)^{-\beta}.
\]
\end{lemma}
A proof of Lemma \ref{lemma:plb_sum} is given in more general terms in Brach et al. \cite[Lemma 3.2]{PLBNetworks}. We use Theorem \ref{thm:average_degree} and Lemma \ref{lemma:plb_sum} to prove the following theorem.
\begin{theorem}
\label{lem:mis-mo}
Let $G=(V, E)$ be a connected \PLBU network with parameters $\beta>2$, $t\geq 0$ and with a universal constant $c_1$.  Then the \gsemo finds a $\left ( \frac{2 c_1 (\beta + t - 1)}{(\beta - 1)(\beta - 2)} + 1 \right )$-approximation of the \is after expected $\bigo{n^3}$ fitness evaluations. In particular, the \gsemo is a $\Theta(1)$-approximation algorithm for the \is on \PLBU networks.
\end{theorem}
\begin{proof}
This argument is similar to the one given in Theorem~\ref{GSEMO_DS} and Theorem~\ref{GSEMO_CDS}. Following an argument as in (Phase 1) in Theorem \ref{GSEMO_DS}, and as in Theorem~\ref{GSEMO_CDS}, one can prove that the solution $0^n$ is added to the Pareto front after expected
\[
\sum_{j = 1}^{n}\left (\frac{j}{e(n+1)n}\right )^{-1} = \bigo{n^2 \log n}
\]
fitness evaluations. Note that the empty set is trivially an independent set. Once $0^n$ is reached, the \gsemo iteratively adds nodes $v_j$ as in Theorem \ref{thm:average_degree} to the current solution, until a maximum independent set is found. Again, since the population size is at most $n + 1$, then selecting a particular solution and adding a chosen node to it occurs after expected $\bigo{n^2}$ fitness evaluations. Hence, after expected $\bigo{n^3}$ fitness evaluations the \ea reaches a maximum independent set with fitness at least that of a set $S_k = \{v_1, \dots, v_k\}$, that consists of a sequence of $k$ nodes s.t. $v_j$ has minimum degree in the complete sub-graph induced by the nodes $V\setminus \left ( \bigcup_{i=1}^{j-1} (N(v_i)\cup v_i )\right )$, as in Lemma \ref{lemma:plb_sum}.

We conclude by proving that $S_k = \{v_1, \dots, v_k\}$ yields the desired approximation guarantee, using the \PLBU property. Let $d_i$ be the number of nodes adjacent the node $v_i$. From Lemma \ref{lemma:plb_sum} it holds
\begin{align}
\sum_{v \in V} \delta (v) & \leq 2 c_1 n (t + 1)^{\beta-1} \sum_{i = 1}^{\Delta(G)} i(i + t)^{-\beta} \nonumber\\
& \leq  2 c_1 n (t + 1)^{\beta-1}  \int_{1}^{\Delta(G)}x(x+t)^{-\beta}dx  \leq  \frac{2 c_1 n (\beta + t - 1)}{(\beta - 1)(\beta - 2)}. \nonumber 
\end{align}
We combine this chain of inequalities with Theorem \ref{thm:average_degree} to conclude that
\begin{equation}
\absl{S_k} \geq \absl{V} \left [ \frac{\sum_{v\in V}\delta(v)}{\absl{V}} + 1 \right ]^{-1} \geq n \left [ \frac{2 c_1 (\beta + t - 1)}{(\beta - 1)(\beta - 2)} + 1 \right ]^{-1}.\nonumber
\end{equation}
Hence, it follows that 
\begin{equation}
\frac{\absl{\opt}}{\absl{S_k}}\leq \frac{n}{\absl{S_k}}\leq \frac{2 c_1 (\beta + t - 1)}{(\beta - 1)(\beta - 2)} + 1,\nonumber
\end{equation}
and the claim follows.
\end{proof}
%
%
%
%
\section{Conclusion}
In this paper, we look at the approximation ratio and run time analysis of commonly studied evolutionary algorithms, for well known \NP-complete covering problems, on power-law bounded networks with exponent $\beta>2$ (see Definition \ref{def:plbu}). We consider the single-objective \ea (see Algorithm \ref{1+1}) and the multi-objective \gsemo (see Algorithm \ref{GSEMO}). 

We prove in Section \ref{sec:ds}, that the \ea reaches a constant-factor approximation ratio for the Minimum Dominating Set problem within $\bigo{n \log n}$ fitness evaluations on power-law bounded networks. Furthermore, we obtain and improved approximation guarantee for the \gsemo. In Section \ref{sec:ds}, we discuss similar bounds for the Minimum Vertex Cover problem. We show that the \ea and the \gsemo reach constant-factor approximation ratios in expected polynomial fitness evaluations for the Minimum Connected Dominating Set and Maximum Independent Set problems, in Section \ref{sec:MCDSs-objsec} and Section \ref{sec:mis}.

In all cases, we observe that the \ea and the  \gsemo reach a constant-factor approximation ratio in polynomial time. This suggests that EAs implicitly exploit the topology of real-world networks to reach a better solution quality than theoretically predicted in the general case.

\bibliography{bibliography}

\begin{thebibliography}{10}
\expandafter\ifx\csname url\endcsname\relax
  \def\url#1{\texttt{#1}}\fi
\expandafter\ifx\csname urlprefix\endcsname\relax\def\urlprefix{URL }\fi
\expandafter\ifx\csname href\endcsname\relax
  \def\href#1#2{#2} \def\path#1{#1}\fi

\bibitem{DJW02}
S.~Droste, T.~Jansen, I.~Wegener, On the analysis of the (1+1) evolutionary
  algorithm, Theoretical Computer Science 276 (2002) 51--81.

\bibitem{JW01}
T.~Jansen, I.~Wegener, Evolutionary algorithms - how to cope with plateaus of
  constant fitness and when to reject strings of the same fitness, {IEEE}
  Transactions on Evolutionary Computation 5 (2001) 589--599.

\bibitem{OW03}
O.~Giel, I.~Wegener, Evolutionary algorithms and the maximum matching problem,
  in: Proc. of {STACS}, 2003, pp. 415--426.

\bibitem{Neumann07}
F.~Neumann, Expected runtimes of a simple evolutionary algorithm for the
  multi-objective minimum spanning tree problem, European Journal of
  Operational Research 181 (2007) 1620--1629.

\bibitem{NW07}
F.~Neumann, I.~Wegener, Randomized local search, evolutionary algorithms, and
  the minimum spanning tree problem, Theoretical Computer Science 378 (2007)
  32--40.

\bibitem{OHY08}
P.~S. Oliveto, J.~He, X.~Yao, Analysis of population-based evolutionary
  algorithms for the vertex cover problem, in: Proc. of {CEC}, 2008, pp.
  1563--1570.

\bibitem{OHY07}
P.~S. Oliveto, J.~He, X.~Yao, Evolutionary algorithms and the vertex cover
  problem, in: Proc. of {CEC}, 2007, pp. 1870--1877.

\bibitem{FHN07}
T.~Friedrich, N.~Hebbinghaus, F.~Neumann, J.~He, C.~Witt, Approximating
  covering problems by randomized search heuristics using multi-objective
  models, in: Proc. of {GECCO}, 2007, pp. 797--804.

\bibitem{BK94}
T.~B\"ack, S.~Khuri, An evolutionary heuristic for the maximum independent set
  problem, in: Proc. of {WCCI}, 1994, pp. 531--535.

\bibitem{Peng15}
X.~Peng, Performance analysis of {(1+1) EA} on the maximum independent set
  problem, in: Proc. of {ICCCS}, 2015, pp. 448--456.

\bibitem{KumarRRT99}
R.~Kumar, P.~Raghavan, S.~Rajagopalan, A.~Tomkins, Trawling the web for
  emerging cyber-communities, Computer Networks 31 (1999) 1481--1493.

\bibitem{LKF05}
J.~Leskovec, J.~Kleinberg, C.~Faloutsos, Graphs over time: Densification laws,
  shrinking diameters and possible explanations, in: Proc. of {KDD}, 2005, pp.
  177--187.

\bibitem{FaloutsosFF99}
M.~Faloutsos, P.~Faloutsos, C.~Faloutsos, On power-law relationships of the
  internet topology, in: Proc. of {SIGCOMM}, 1999, pp. 251--262.

\bibitem{Newman03}
M.~E.~J. Newman, The structure and function of complex networks, {SIAM} Review
  45 (2003) 167--256.

\bibitem{WS98}
D.~J. Watts, S.~H. Strogatz, Collective dynamics of 'small-world' networks,
  Nature 393 (1998) 440--442.

\bibitem{PLBNetworks}
P.~Brach, M.~Cygan, J.~{\L}{{a}}cki, P.~Sankowski, Algorithmic complexity of
  power law networks, in: Proc. of {SODA}, 2016, pp. 1306--1325.

\bibitem{DBLP:journals/algorithmica/DoerrJW12}
B.~Doerr, D.~Johannsen, C.~Winzen, Multiplicative drift analysis, Algorithmica
  64~(4) (2012) 673--697.

\bibitem{GK98}
S.~Guha, S.~Khuller, Approximation algorithms for connected dominating sets,
  Algorithmica 20 (1998) 374--387.

\bibitem{KMS98}
S.~Khanna, R.~Motwani, M.~Sudan, U.~Vazirani, On syntactic versus computational
  views of approximability, {SIAM} Journal on Computing 28 (1998) 164--191.

\bibitem{DBLP:journals/algorithmica/HalldorssonR97}
M.~M. Halld{\'{o}}rsson, J.~Radhakrishnan, Greed is good: Approximating
  independent sets in sparse and bounded-degree graphs, Algorithmica 18~(1)
  (1997) 145--163.

\end{thebibliography}

\end{document}